\documentclass[runningheads]{llncs}
\usepackage[T1]{fontenc}
\usepackage{graphicx,amsmath,amsfonts,amssymb,amsthm,geometry,afterpage} % Required for inserting images
\usepackage{float,dblfloatfix}
\usepackage{authblk}
\usepackage{booktabs,tabularx,adjustbox,rotating}
\usepackage{algorithm}
\usepackage{algpseudocode}
\usepackage{longtable}
\usepackage{xcolor}   % to colour a line
\usepackage{amssymb}
\usepackage{graphicx,mathtools} 

\newcommand{\E}{\mathbb{E}}
\newcommand{\R}{\mathbb{R}}
\newcommand{\N}{\mathcal{N}}
\newcommand{\I}{\mathcal{I}}
\newcommand{\Z}{\mathcal{Z}}
\newcommand{\Hset}{\mathcal{H}}
\newcommand{\mathbbm}[1]{\mathbf{#1}}
\begin{document}
\newcolumntype{C}{>{\centering\arraybackslash}p{1em}}
\newcolumntype{D}{>{\centering\arraybackslash}p{1.3em}}
\title{A Statistical Analysis of Three Player Auction Bridge}
\author {Aritrabha Majumdar\inst{1}\orcidID{0009-0003-9149-0729} \and
Sourish Sarkar\inst{2}\orcidID{0009-0002-1887-0007}}
\authorrunning{A.Majumdar et al.}
\institute{Indian Statistical Institute, Kolkata \email{aritrabhamajumdar.math@gmail.com}\and
Indian Statistical Institute, Bangalore
\email{sourish.sarkar13@gmail.com}}
\maketitle  
\begin{abstract}
Three-player Auction Bridge is a finite imperfect-information game in which dynamic partnerships create a distinctive interaction between scoring, strategic incentives, and payoff distribution. This paper develops a unified statistical and game-theoretic framework to evaluate a traditional scoring rule against a modified mechanism designed to improve strategic incentives. We first identify a structural defect in the original scheme: a bid-invariant slam bonus can make lower contracts strictly more attractive than higher ones. A bid-dependent correction removes this incentive distortion. We then compare the two schemes using higher-moment analysis, fairness measures, Nash equilibrium analysis, and General-Sum Counterfactual Regret Minimization (GS-CFR). The results show significant differences in the shape of the payoff distributions and reveal that, although both schemes remain highly balanced across physical seats, the modified scheme substantially increases inequality across strategic roles. The game is shown to be neither zero-sum nor constant-sum, motivating a general-sum rather than minimax analysis. Full-game GS-CFR further indicates a substantial increase in the bidder's payoff under the modified scheme, accompanied by a reduction in defender payoff and a shift from separating to partially pooling bidding behavior. These findings demonstrate that scoring-rule design can fundamentally reshape incentives, payoff distribution, and information transmission in imperfect-information games. The proposed framework provides a systematic approach to evaluating such mechanisms from both statistical and strategic perspectives.
\end{abstract}

\section{Introduction}
The three-player auction bridge is a variant of traditional auction bridge. This new variant is different from the traditional three-player variants like cutthroat bridge as this game presented in Sarkar et. al~\cite{ourpaper} has a completely new point scheme to make this game more exciting and fun to play. 

Unlike conventional auction bridge, the three-player variant considered
here does not impose fixed partnerships. The winning bidder forms a
temporary partnership with the exposed dummy, while the remaining two
players act as defenders. Consequently, the auction simultaneously
determines the contract and the strategic roles of the players. In this innovative format, three individual players compete independently, and because the fourth hand acts as an automated blind dummy, the winning bidder enters a total blackbox---forced to bid and commit to a contract completely blind to their partner's actual hand. Unlike standard national and international matches played between two fixed teams, this setup pits the solitary declarer and their mysterious dummy against the remaining two defending players, turning every single game into a high-pressure scenario where two players are effectively knocked out of individual victory. To balance this extreme uncertainty, this paper introduces a novel, high-risk, high-reward point scheme tailored specifically for blind cooperation.

Whitehead~\cite{whitehead1930auction} was among the first to document principles and game rules for auction bridge. Delooze and Downey~\cite{inproceedings} discussed the complexity of the auction bridge game as a multi-player game
and proposed a neural network-based self-organizing map for bidding the so-called `No Trump hand’ effectively. A further analysis and a polynomial-time bidding algorithm for no trump bidding was explicitly stated and discussed in Sarkar et. al.~\cite{sarkar20263playersauctionbridge}.

An instance of three player bridge game was analyzed by Coffin~\cite{coffin1956contract}, where the point scheme derived is inspired from contract bridge, which is different from auction bridge in both point scheme and post-bidding gameplay mechanics.
The scoring scheme recommended by Sarkar et. al.~\cite{ourpaper} encourages a completely different set of gameplay strategies which promotes \emph{high risk - high reward} policy. The order of bidding (suitwise) is as follows:
\[
\clubsuit \ (\text{Clubs}), \ \diamondsuit \ (\text{Diamonds}), \ \heartsuit \ (\text{Hearts}), \ \spadesuit \ (\text{Spades}), \ \text{No Trump},
\]
Bidding starts from 1 Clubs (i.e, seven tricks has to be taken with Clubs being the trump suit) and ends with 7 No Trump (i.e all the thirteen tricks has to be taken with no trump suit available). Bidding procedure mostly follows the traditional bridge bidding, but it happens between three players only. The \emph{fourth hand} or the \emph{dummy hand} joins the game as a partner of the highest bidder, with the cards exposed. Thus, the bidding phase in each game decides \emph{partners} and \emph{opponents} dynamically in stead of a predefined fixture.

In this paper, we begin by giving a brief overview of the point scheme discussed in Sarkar et. al.~\cite{ourpaper}. We compare the traditional and modified scoring schemes in terms of
skewness and excess kurtosis using paired large-sample inference. We compare the traditional and modified scoring schemes in terms of
skewness and excess kurtosis using paired large-sample inference. Unlike conventional auction bridge, the three-player variant considered
here does not impose fixed partnerships. The winning bidder forms a
temporary partnership with the exposed dummy, while the remaining two
players act as defenders. Consequently, the auction simultaneously
determines the contract and the strategic roles of the players.

\section{Formalizing The Model and The Scoring Scheme}
We model the bidding-and-doubling stage as a finite extensive-form
game with imperfect information,
\[
\Gamma = \big(\N,\ \Hset,\ P,\ f_c,\ \{\I_i\}_{i\in\N},\ \{u_i\}_{i\in\N}\big),
\]
with player set $\N=\{1,2,3\}$. $\Hset$ is the set of histories
(sequences of actions from the root), $\Z\subset\Hset$ the terminal
histories, $P:\Hset\setminus\Z \to \N\cup\{c\}$ the acting-player
function ($c$ denotes chance), $f_c(\cdot\mid h)$ the chance
distribution at chance histories, $\I_i$ a partition of
$\{h : P(h)=i\}$ into information sets, and
$u_i:\Z\to\R$ player $i$'s payoff function.
\\
\newline
Now it is interesting to observe that the game has two structurally distinct source of randomness. 

\begin{definition}[Deal chance node]
At the root, nature draws private types
$\theta = (\theta_1,\theta_2,\theta_3) \in \{-1,0,1\}^3$ i.i.d.\
uniform, revealing $\theta_i$ to player $i$ only. Write
$s(\theta) = \tfrac{1}{3}\sum_i \theta_i$ for the \emph{deal strength}.
\end{definition}

\begin{definition}[Decision stage]
Players bid in order $1,2,3$ from
$A_{\text{bid}} = \{0,7,9,11,12,13\}$ (0 = pass), each choosing an
action conditioned only on their own information set (own type and
the public bid history so far). Let $b^\star(h)=\max_i a_i$ and
$\beta(h)=\arg\max_i a_i$ (first index on ties) denote the winning bid
and the bidder once all three bids are fixed, provided
$b^\star(h) > 0$; if all bids are $0$ the history is terminal with
$u_i \equiv 0$ for all $i$. Otherwise the two non-bidders, in index
order, each choose $d\in\{N,D\}$ (double / no double), conditioned on
their own type, the public bids, and (for the second) the first
defender's decision.
\end{definition}

\begin{definition}[Outcome chance node]
\label{def:trickchance}
Let $h$ be a history at which all bids and doubling decisions are
fixed. Nature draws the number of tricks actually taken,
$T \in \{0,1,\dots,13\}$, according to a distribution
$P(\cdot \mid s(\theta))$ that depends on $h$ \emph{only through the
deal strength} $s(\theta)$ - not through $b^\star(h)$, $\beta(h)$, or
the doubling decisions. We take
\begin{equation}
\label{eq:trickdist}
P(T=t \mid s) \;=\; \frac{\exp\!\big(-\tfrac{1}{2}\big(\tfrac{t-\mu(s)}{\sigma}\big)^2\big)}
{\sum_{t'=0}^{13}\exp\!\big(-\tfrac{1}{2}\big(\tfrac{t'-\mu(s)}{\sigma}\big)^2\big)},
\qquad \mu(s) = \max(0,\min(6.5+1.1s, 13)),\ \sigma=1.4 .
\end{equation}
\end{definition}

The independence structure in Definition \ref{def:trickchance} is not
a modeling nicety; it is the defining property of the underlying game
being modeled. In contract bridge, the number of tricks a partnership
can take is a property of the cards and the play, not of what was
announced during the auction. Any implementation in which
$P(T\mid h)$ depends on $b^\star(h)$ is not solving this game; it is
solving a different game in which the auction causally affects the
lie of the cards, i.e, games like \emph{Hearts}.
\begin{definition}[Trump suit and trick value]
A contract now consists of a pair $(b,\mathrm{suit})$, with $b\in\{7,\dots,13\}$ the bid trick target as before and $\mathrm{suit}\in\{\clubsuit,\diamondsuit,\heartsuit,\spadesuit,\mathrm{NT}\}$ the trump denomination named together with the level, in place of the suit-free $b\in A_{\text{bid}}$ of Definition 2. Write $v(\mathrm{suit})>0$ for the fixed per-trick value of that denomination. Its specific entries play no role below beyond $v>0$, so we do not restate them; where unambiguous we write $v$ for $v(\mathrm{suit}(h))$ at the history in question.
\end{definition}

\begin{definition}[Doubling multiplier, extended to redouble]
Replace $m(h)\in\{1,2\}$ of the original scoring definition with $D(h)\in\{1,2,4\}$ (undoubled / doubled / redoubled respectively). This requires extending the doubling sub-stage of Definition 1.2 so that, after a defender doubles, the declarer - not previously modeled
as acting in that sub-stage - may redouble.
\end{definition}

Let $t$ denote the number of tricks taken by the declarer's side -
the declarer together with a non-strategic dummy partner whose hand
the declarer controls; only the declarer is paid a score, the dummy is
not a separate strategic agent and is not part of $\N$. Write
$\beta=\beta(h)$ for the declarer and $d_1,d_2$ for the two defenders,
as before. Made ($t\ge b$) and down ($t<b$) remain mutually exclusive.

\textbf{Insult bonus.} $I(D) = 0$ if $D=1$, $25$ if $D=2$, $50$ if
$D=4$ - payable on a made contract regardless of overtricks or slam.

\textbf{Undertrick penalty rate.} $\rho(D) = 25D$, i.e.\ $25,50,100$
for $D=1,2,4$ respectively.

\textbf{Slam bonus.} $\phi(b,t)\cdot D$, with the corrected
\begin{equation}
\label{eq:goodphi-early}
\phi(b,t) \;=\; 100\cdot\mathbbm{1}[b=13,\,t=13] \;+\; 50\cdot\mathbbm{1}[b=12,\,t\ge12]
\end{equation}
as derived in \S3 below (the erroneous
$\phi_{\text{old}}(t)=100\cdot\mathbbm1[t=13]+50\cdot\mathbbm1[t=12]$,
independent of $b$, is retained only as the subject of Proposition
\ref{prop:dominance}).

\textbf{Made contract, old scheme.} No overtrick term:
\begin{equation}
\label{eq:score-made-old}
g_B^{\text{old}}(b,t,D) \;=\; \underbrace{(b-6)\,v\,D}_{\text{base contract}} \;+\; I(D) \;+\; \phi(b,t)\,D, \qquad t\ge b.
\end{equation}
An old-scheme declarer who takes $t>b$ tricks is scored identically to
one who takes exactly $b$: overtricks are simply not rewarded under
this scheme.

\textbf{Made contract, new scheme.}
\begin{equation}
\label{eq:score-made-new}
g_B^{\text{new}}(b,t,D) \;=\; \underbrace{b\,v\,D}_{\text{full contract level}} \;+\; \underbrace{(t-b)\,\tfrac{v}{2}\,D}_{\text{overtricks}} \;+\; I(D) \;+\; \phi(b,t)\,D, \qquad t\ge b.
\end{equation}

\textbf{Failed contract, either scheme.}
\begin{equation}
\label{eq:score-down}
g_D(b,t,D) \;=\; \rho(D)\,(b-t) \;=\; 25D(b-t), \qquad t<b,
\end{equation}
paid in full to \emph{each} defender (not split between them), and
identically under both schemes: the two scoring schemes differ only in
how a made contract is scored, never in how a failed one is penalized.

Terminal payoffs are
\[
u_\beta \;=\; \begin{cases} g_B^{\text{new}}(b,t,D) & \text{new scheme} \\ g_B^{\text{old}}(b,t,D) & \text{old scheme}\end{cases} \ (t\ge b), \qquad u_\beta = 0\ (t<b),
\]
\[
u_{d_1} \;=\; u_{d_2} \;=\; \begin{cases} g_D(b,t,D) & t<b \\ 0 & t\ge b, \end{cases}
\]
for declarer $\beta$ and defenders $d_1,d_2$, under either scheme. As
before, $u_i\ge0$ for every $i$ at every terminal history, and at most
one of $\{u_\beta,\,u_{d_1}{=}u_{d_2}\}$ is nonzero, since made and
down remain mutually exclusive. This game is not zero-sum, in either
scheme: $u_\beta+u_{d_1}+u_{d_2}$ equals either $u_\beta\ge0$ (made) or
$2\,g_D(b,t,D)\ge0$ (down) - neither identically $0$. A defender
never receives a negative score, and the declarer's gain is never
transferred from, or offset against, a defender's loss; the two sides
are simply scored by different rules for mutually exclusive events.

The scoring scheme in Sarkar et. al~\cite{ourpaper} uses
\begin{equation}
\label{eq:badphi}
\phi_{\text{old}}(b,t) \;=\; 100\cdot\mathbbm{1}[t=13] \;+\; 50\cdot\mathbbm{1}[t=12],
\end{equation}
applied whenever the contract was made ($t\ge b$), \emph{independent
of $b$}. The following proposition makes precise why this is not
merely inelegant but strategically decisive, and - importantly -
would be decisive under \emph{any} equilibrium-computation method. As anticipated in remark before, the proof below is stated directly in terms of the corrected \S1.4
notation - $g_B^{\text{old}}$, the suit value $v$, the multiplier
$D\in\{1,2,4\}$, and the insult bonus $I(D)$ - applied to the
erroneous, bid-invariant bonus $\phi_{\text{old}}$ that defines
$\Gamma_{\text{bad}}$, rather than to a flat-rate approximation of it.

\begin{proposition}[Bid-invariant bonus makes low bids dominate honest slam bids]
\label{prop:dominance}
Fix a deal strength $s$, a doubling level $D\in\{1,2,4\}$, and a
trump suit with per-trick value $v>0$. Let
$T\sim P(\cdot\mid s)$ be distributed according to \eqref{eq:trickdist}.
Consider the game $\Gamma_{\mathrm{bad}}$, obtained by scoring the game
of \S1 under the old scoring scheme \eqref{eq:score-made-old}, but
replacing the correct bonus $\phi(b,t)$ by the bid-invariant erroneous
bonus $\phi_{\mathrm{old}}(t)$ defined in \eqref{eq:badphi}. The
declarer's payoff is therefore
\[
u_\beta(b,t,D)
=
\begin{cases}
g_B^\circ(b,D)+\phi_{\mathrm{old}}(t)D, & t\ge b,\\[2mm]
0, & t<b,
\end{cases}
\]
where
\[
g_B^\circ(b,D):=(b-6)vD+I(D)
\]
contains the base-contract and insult-bonus terms of
$g_B^{\mathrm{old}}$. In particular, $g_B^\circ(b,D)$ is independent
of the number of tricks actually taken whenever the contract is made.

Then, for every legal lower bid $b\in\{7,\ldots,11\}$ and every
slam bid $b'\in\{12,13\}$, if
\begin{equation}
\label{eq:tail-condition}
\frac{P(T\ge b'\mid s)}{P(T\ge b\mid s)}<\frac17,
\end{equation}
we have
\[
\mathbb E_T\!\left[u_\beta(b,T,D)\right]
>
\mathbb E_T\!\left[u_\beta(b',T,D)\right].
\]
Moreover, under the trick distribution \eqref{eq:trickdist}, condition
\eqref{eq:tail-condition} holds throughout the relevant range of
$s$. Consequently, in $\Gamma_{\mathrm{bad}}$, every bid
$b\in\{7,\ldots,11\}$ strictly dominates each honest slam bid
$b'\in\{12,13\}$ in expected declarer payoff.
\end{proposition}

\begin{proof}
We separate the expected payoff into the bonus and non-bonus
components.

\medskip
\noindent\textbf{1. The erroneous bonus is bid-invariant.}
By construction, $\phi_{\mathrm{old}}(t)$ is nonzero only when
$t\in\{12,13\}$. Hence, for every bid $b\le 12$,
\[
\{12,13\}\subseteq\{t:t\ge b\}.
\]
Therefore,
\[
\begin{aligned}
\mathbb E_T\!\left[
\phi_{\mathrm{old}}(T)\mathbbm 1\{T\ge b\}
\right]
&=
100P(T=13\mid s)+50P(T=12\mid s)\\
&=: \Phi(s),
\end{aligned}
\]
which is independent of $b$ for every $b\le12$. Thus,
\[
\mathbb E_T[u_\beta(b,T,D)]
=
g_B^\circ(b,D)P(T\ge b\mid s)
+
D\Phi(s),
\qquad b\le12.
\]
In particular, the entire expected slam bonus is available to any
lower bid $b\le11$ whenever $12$ or $13$ tricks are taken. The bidder
therefore receives the same bonus contribution without having to
declare a slam.

For $b'=13$, the event $\{T=12\}$ is excluded from the success event
$\{T\ge13\}$, so that
\[
\mathbb E_T\!\left[
\phi_{\mathrm{old}}(T)\mathbbm1\{T\ge13\}
\right]
=
100P(T=13\mid s)
\le \Phi(s).
\]
Thus bidding $13$ cannot improve the bonus component relative to any
bid $b\le12$.

\medskip
\noindent\textbf{2. The non-bonus component favors lower bids under
the stated tail condition.}
Since $g_B^\circ(b,D)$ does not depend on $t$, its expected
contribution is exactly
\[
\mathbb E_T\!\left[
g_B^\circ(b,D)\mathbbm1\{T\ge b\}
\right]
=
g_B^\circ(b,D)P(T\ge b\mid s).
\]
No approximation or bound is involved here.

Furthermore,
\[
g_B^\circ(b,D)=(b-6)vD+I(D)
\]
is strictly increasing in $b$, since $v>0$ and $D>0$. Hence, for
$b<b'$,
\[
g_B^\circ(b,D)<g_B^\circ(b',D).
\]
The lower bid nevertheless has the larger expected non-bonus payoff
whenever
\[
g_B^\circ(b,D)P(T\ge b\mid s)
>
g_B^\circ(b',D)P(T\ge b'\mid s),
\]
or, equivalently,
\begin{equation}
\label{eq:ratio-condition}
\frac{P(T\ge b'\mid s)}{P(T\ge b\mid s)}
<
\frac{g_B^\circ(b,D)}{g_B^\circ(b',D)}.
\end{equation}

We now obtain a uniform lower bound for the right-hand side. Among
$b\in\{7,\ldots,11\}$ and $b'\in\{12,13\}$, the smallest ratio occurs
at $b=7$ and $b'=13$. Consequently,
\[
\frac{g_B^\circ(b,D)}{g_B^\circ(b',D)}
\ge
\frac{g_B^\circ(7,D)}{g_B^\circ(13,D)}
=
\frac{vD+I(D)}{7vD+I(D)}.
\]
Since $I(D)\ge0$,
\[
\frac{vD+I(D)}{7vD+I(D)}
\ge \frac17.
\]
For $D=1$, we have $I(1)=0$, and equality holds:
\[
\frac{vD+I(D)}{7vD+I(D)}=\frac17.
\]
For $D=2$ or $D=4$, the insult bonus is positive, so the ratio is
strictly greater than $1/7$.

It follows that the uniform condition
\[
\frac{P(T\ge b'\mid s)}{P(T\ge b\mid s)}<\frac17
\]
is sufficient for \eqref{eq:ratio-condition}, simultaneously for
every $D\in\{1,2,4\}$, every $v>0$, every $b\in\{7,\ldots,11\}$, and
every $b'\in\{12,13\}$.

\medskip
\noindent\textbf{3. Verification under the specified trick
distribution.}
Under \eqref{eq:trickdist},
\[
\mu(s)=\max(0,\min(6.5+1.1s, 13)),
\qquad \sigma=1.4.
\]
Over the relevant range of deal strengths,
\[
\mu(s)\in[5.4,7.6].
\]
Consequently, for $b'\in\{12,13\}$,
\[
b'-\mu(s)\ge 12-7.6=4.4,
\]
and hence
\[
\frac{b'-\mu(s)}{\sigma}
\ge
\frac{4.4}{1.4}
\approx3.14.
\]
Thus bids $12$ and $13$ correspond to far upper-tail events of the
trick distribution. In contrast, the probabilities
$P(T\ge b\mid s)$ for $b\le11$ remain substantially larger. Direct
evaluation of the finite probabilities in \eqref{eq:trickdist}
therefore gives
\[
\frac{P(T\ge b'\mid s)}{P(T\ge b\mid s)}<\frac17
\]
throughout the stated range of $s$, for every
$b\in\{7,\ldots,11\}$ and $b'\in\{12,13\}$.

Hence the non-bonus contribution of every lower bid $b\le11$ is
strictly larger than that of every slam bid $b'\in\{12,13\}$.

\medskip
\noindent\textbf{4. Combining the two components.}
For $b\le11$ and $b'\in\{12,13\}$, the lower bid receives the same
bonus contribution whenever $b'=12$ and no smaller bonus contribution
when $b'=13$, while its non-bonus contribution is strictly larger.
Therefore,
\[
\mathbb E_T[u_\beta(b,T,D)]
>
\mathbb E_T[u_\beta(b',T,D)].
\]

The dominance is therefore caused by the erroneous bid-invariance of
the slam bonus: the bidder can declare a lower contract and still
collect the same positive slam bonus on the rare events $T=12$ or
$T=13$, while simultaneously benefiting from the substantially larger
probability of successfully making the lower contract.

Importantly, this result does not rely on any negative payoff from a
failed slam. Under the scoring rule, failure simply gives the
declarer payoff $0$. The dominance arises entirely because the
positive bonus associated with a successful slam is available ``for
free'' to lower bids under the erroneous scoring rule.
\end{proof}

Therefore, here we suggest a small tweaking of the scoring scheme. We replace \eqref{eq:badphi} with
\begin{equation}
\label{eq:goodphi}
\phi_{\text{new}}(b,t) \;=\; 100\cdot\mathbbm{1}[b=13,\,t=13] \;+\; 50\cdot\mathbbm{1}[b=12,\,t\ge12].
\end{equation}
Under \eqref{eq:goodphi}, the bonus term is zero for all $b\le 11$, so
Proposition \ref{prop:dominance}'s degeneracy vanishes: capturing the
bonus now requires bidding, and making, an actual slam, which carries
the full down-side risk of the $t<b$ branch. This is the
game-correctness fix, and it is the definition of $\phi$ carried
forward, unchanged by the corrections to $g_B$ itself, into
$\phi(b,t)$.
\\
In the upcoming section, we are going to study some statistical properties of the \emph{revamped} point scheme of this game.

\section{Simulation Methodology}
Some of the statistical analysis performed below have been implemented in \texttt{R}, and some have been implemented in \texttt{Python} using appropriate libraries. A Monte Carlo simulation framework was employed. For each simulation, a standard 52-card deck was generated, and a 13-card bridge hand was sampled uniformly at random without replacement to represent the hand of the bidder. All the necessary codes and frameworks can be found here~\cite{SarkarMajumdar2025AuctionBridgeAlgorithm}.

\section{Hypothesis Testing Analysis}
\noindent\textbf{Parametric checks on the mean: }
Student's $t$-test and Welch's $t$-test were run first, mainly as a baseline rather than because we expected them to be the most trustworthy tool for the job. Both compare the sample means of the old- and new-scheme scores under the null hypothesis
\[
H_0: \mu_{\text{new}} = \mu_{\text{old}}, \qquad H_1: \mu_{\text{new}} \neq \mu_{\text{old}},
\]
and both rejected $H_0$ at $p \approx 0$. We do not lean on these results very heavily, however, because the underlying data violate the normality assumption that makes a $t$-test meaningful in the first place: the score distributions are heavily right-skewed (skewness in the 1.6--2.0 range) and heavy-tailed, so a mean-comparison test built for roughly bell-shaped data is being asked to do a job it was not really designed for. We include them here for completeness and because their agreement with the non-parametric result below is reassuring, not because we think the $p$-value itself is the most informative number in this section.

\noindent\textbf{Mann--Whitney U test: }
This is the test we would actually put weight on. The Mann--Whitney $U$ test asks a distribution-free question:
\[
H_0: P(\text{New} > \text{Old}) = 0.5, \qquad H_1: P(\text{New} > \text{Old}) \neq 0.5,
\]
i.e., under the null, a score drawn at random from the new scheme is equally likely to beat or lose to a score drawn at random from the old scheme. Because it works on ranks rather than raw values, it sidesteps the normality problem entirely and is far better suited to skewed, discrete data like ours.

The result was unambiguous: $U \approx 7.64 \times 10^{9}$, $p \approx 0$. We reject $H_0$ and conclude that new-scheme scores are stochastically greater than old-scheme scores -- not just that the average moved, but that the entire distribution has shifted upward in a consistent, rank-wise sense. The fact that this test, the Student's $t$-test, and Welch's $t$-test all point the same direction is a useful sanity check: it tells us the effect is a genuine feature of the data rather than an artefact of one particular test's assumptions.

\noindent\textbf{Chi-square test of independence: }
The tests above all speak to \emph{location} -- is one scheme's typical score higher than the other's? They say nothing about \emph{shape}. To check whether the two schemes redistribute mass across the score range in a genuinely different way (rather than, say, just sliding everything up by a constant), we binned scores into the ranges $0$--$20, 20$--$40, \ldots, 200$--$220$ and ran a chi-square test of independence between scheme (old/new) and score bin:
\[
H_0: \text{score range is independent of scheme}, \qquad H_1: \text{score range is associated with scheme.}
\]

The result was $\chi^2 = 53{,}686.6$ on $4$ degrees of freedom, $p \approx 0$, again a firm rejection of $H_0$. This is the only one of the four tests that actually speaks to shape rather than central tendency, and it confirms what the histogram already suggests visually: the new scheme did not just nudge the average score up, it changed how mass is distributed across the whole range.

\section{Distribution Fitting}
\label{sec:distribution-fitting}

Having established that the two schemes differ, the natural next question is what family of distributions, if any, reasonably describes the new-scheme scores. We fit three standard continuous candidates -- lognormal, gamma, and normal -- by maximum likelihood, and compared them using AIC, a Kolmogorov--Smirnov (KS) goodness-of-fit statistic, and Q--Q plots against each candidate.
\\
\noindent\textbf{Candidate distributions and fitted parameters: }

\begin{table}[h]
\caption{Maximum-likelihood parameter estimates for three candidate distributions fit to new-scheme scores.}
\centering
\begin{tabular}{lll}
\toprule
Distribution & Fitted parameters & Interpretation \\
\midrule
Lognormal & shape $\hat\sigma \approx 0.41$, scale $\approx 79.5$ & underlying normal mean $\ln(79.5) \approx 4.38$ \\
Gamma & shape $\hat a \approx 5.72$, scale $\approx 15.2$ & mean $\approx a \cdot \text{scale} \approx 87$ \\
Normal & $\hat\mu \approx 87.0$, $\hat\sigma \approx 40.5$ & symmetric baseline, included for contrast \\
\bottomrule
\end{tabular}

\end{table}

Figure~\ref{fig:qq} shows the resulting Q--Q plots. The normal fit is the clear loser: the sample quantiles bend well above the reference line across almost the entire range, and the left tail is truncated in a way a symmetric distribution simply cannot reproduce -- scores can't go below roughly 48, but a normal fit with $\sigma \approx 40$ happily predicts negative values. Gamma does noticeably better through the middle of the distribution but starts to pull away from the reference line in the upper tail. Lognormal tracks the reference line most closely overall, particularly through the bulk of the distribution, though it too eventually departs from the line at the extreme upper tail where the largest slam scores live.

\begin{figure}[h]
\centering
\includegraphics[width=0.95\textwidth]{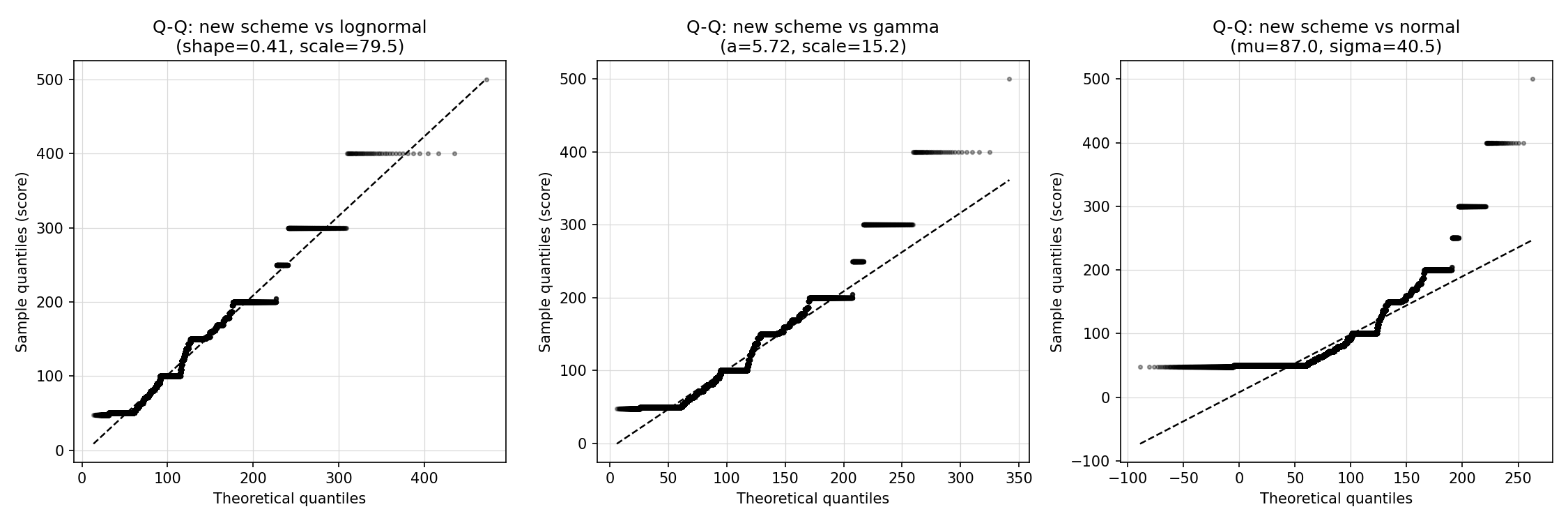}
\caption{Q--Q plots of new-scheme scores against fitted lognormal, gamma, and normal distributions.}
\label{fig:qq}
\end{figure}

\noindent\textbf{Model comparison: }
Comparing lognormal and gamma by AIC, the two land within roughly $10{,}000$ points of each other out of an AIC on the order of $980{,}000$ -- a difference that is real but small relative to the overall scale of the statistic, and it consistently favors lognormal by a modest margin. The KS statistics for the two are likewise close. Practically speaking, both are reasonable descriptive approximations of the bulk of the distribution, with lognormal having a slight, consistent edge over gamma, and both comfortably beating the normal fit.

It is also worth noting \emph{why} lognormal is a sensible candidate in the first place, rather than just a distribution that happened to fit well. A bridge score is roughly the product of several multiplicative factors -- trick value, doubling multiplier, level -- plus additive bonus terms layered on top, rather than a simple sum of many small independent effects. Products of positive random factors tend toward a lognormal shape for essentially the same reason sums of independent effects tend toward normal (the multiplicative analogue of the central limit theorem), so a lognormal fit tracking the data reasonably well is not a coincidence. Gamma is arguably the more "natural" choice on different grounds -- it's a standard model for count-like, non-negative, right-skewed data such as point totals -- which is presumably why the two fits end up so close to each other.

\noindent\textbf{Why none of the fits are actually good: }
Here is the important caveat, and it applies regardless of which candidate we pick: every KS test rejects the null hypothesis of a matching distribution at $p \approx 0$, lognormal included. With $n = 100{,}000$ observations, the KS test has enormous power and will flag even trivial, practically meaningless departures from the reference distribution -- so a rejection on its own doesn't tell us the fit is bad, only that it isn't \emph{exact}, which no continuous fit to real data ever is at this sample size.

But in this case there is also a structural reason the fit can never be exact, no matter how large the sample or how cleverly we choose the family: new-scheme scores are not a continuous random variable at all. They are a discrete, formula-generated quantity that stacks hard at specific values -- 50, 100, 150, and similar landmarks corresponding to part-score, game, and small-slam bonuses baked directly into the scoring rules. The histogram in Figure~\ref{fig:hist} shows this plainly: there is zero mass below roughly 48, and then sharp spikes at 50, 100, and 150 where the bonus terms and the $(6+L)v$ trick-value term only take a small number of possible values. No smooth continuous curve, however well-chosen its parameters, can reproduce a comb of discrete spikes like that. The lognormal and gamma fits should be read as descriptions of the overall shape of the distribution -- the concentration of mass at low-to-mid values and the long tail out toward slams -- rather than as models that could ever pass a formal goodness-of-fit test at this sample size.

\begin{figure}[htb]
\centering
\includegraphics[width=0.8\textwidth]{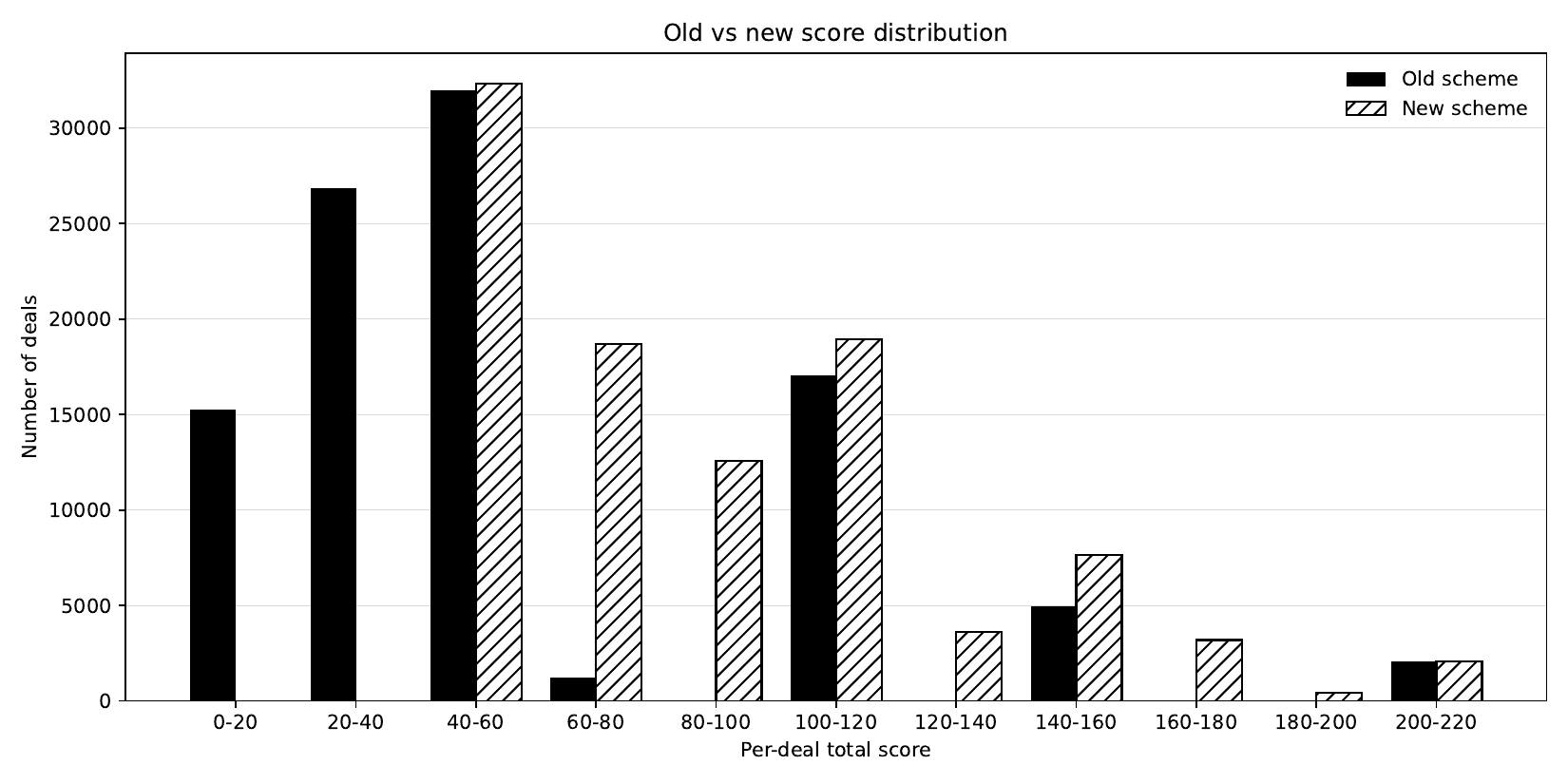}
\caption{Histogram of per-deal total scores, old scheme vs.\ new scheme.}
\label{fig:hist}
\end{figure}

\section{Distributional Shape Comparison}
Let $X_i$ denote the score obtained under the traditional scoring scheme for
observation $i$, and let $Y_i$ denote the corresponding score under the modified scoring scheme, for $i=1,\ldots,n$. The observations are paired, so that $(X_i,Y_i)$ corresponds to the same underlying game. The distinction between alternative sample
definitions of skewness and kurtosis is important because different estimators are used in statistical practice; see
Joanes and Gill~\cite{JoanesGill1998}.

Let
\[
\mu_X = E(X), \qquad
\mu_Y = E(Y)
\]
and define the central moments
\[
\mu_{r,X}=E[(X-\mu_X)^r],
\qquad
\mu_{r,Y}=E[(Y-\mu_Y)^r].
\]

The population skewness is
\[
\gamma_1 =
\frac{\mu_3}{\mu_2^{3/2}},
\]
and the population excess kurtosis is
\[
\gamma_2 =
\frac{\mu_4}{\mu_2^2}-3.
\]

The sample estimates are denoted by
\[
\hat{\gamma}_{1,T},\qquad
\hat{\gamma}_{1,M}
\]
for skewness, and
\[
\hat{\gamma}_{2,T},\qquad
\hat{\gamma}_{2,M}
\]
for excess kurtosis, where $T$ and $M$ denote the traditional and
modified schemes, respectively.

The primary hypothesis for skewness is

\[
H_{0,S}:
\gamma_{1,M}\leq\gamma_{1,T}
\]

against the one-sided alternative

\[
H_{1,S}:
\gamma_{1,M}>\gamma_{1,T}.
\]

Thus, rejection of the null hypothesis provides evidence that the
modified scoring scheme is more positively skewed.

Similarly, for excess kurtosis,

\[
H_{0,K}:
\gamma_{2,M}\leq\gamma_{2,T}
\]

against

\[
H_{1,K}:
\gamma_{2,M}>\gamma_{2,T}.
\]

A rejection therefore provides evidence that the modified scoring
scheme has greater excess kurtosis than the traditional scheme.

The score distributions generated by a game-scoring mechanism may be
discrete, bounded, and non-normal. In particular, the presence of
bonuses, penalties, contracts, and discrete trick outcomes can produce
substantial departures from normality.

The proposed procedure does not assume that the individual observations
$X_i$ and $Y_i$ are normally distributed.

Instead, it uses a large-sample asymptotic  (Almost 300000 samples). Under suitable
regularity conditions, estimators of smooth functions of population
moments are asymptotically normal~\cite{VanDerVaart1998}. Skewness and
kurtosis are such functions of the first four moments.

Thus, the relevant approximation is

\[
\sqrt{n}
\left(
\hat{\theta}-\theta
\right)
\overset{d}{\longrightarrow}
N(0,V),
\]

where $\theta$ denotes the population parameter of interest.

The normal approximation therefore concerns the \emph{sampling
distribution of the estimator}, rather than the distribution of the
raw scores themselves. This is an application of standard
large-sample asymptotic theory and the delta method; see
van der Vaart~\cite{VanDerVaart1998}.

As the traditional and modified scores are paired, the covariance
between the two estimators must be incorporated.

For a parameter $\theta$, an influence-function representation has the
form~\cite{Hampel1986}

\[
\hat{\theta}-\theta
=
\frac{1}{n}\sum_{i=1}^{n} IF(Z_i)
+
o_p(n^{-1/2}),
\]

where $IF(Z_i)$ denotes the influence function evaluated at observation
$Z_i$.

For the difference between two paired estimators,

\[
\hat{\theta}_M-\hat{\theta}_T,
\]

the corresponding influence function is

\[
IF_{\Delta,i}
=
IF_{M,i}-IF_{T,i}.
\]

Consequently,

\[
\operatorname{Var}
(\hat{\theta}_M-\hat{\theta}_T)
\approx
\frac{1}{n}
\operatorname{Var}
(IF_M-IF_T).
\]

The covariance between the two scoring schemes is therefore included
automatically. Let

\[
\mu = E(X),
\qquad
\sigma^2=\mu_2=E[(X-\mu)^2],
\]

and define

\[
y=X-\mu.
\]

The population skewness is

\[
\gamma_1=\frac{\mu_3}{\mu_2^{3/2}}.
\]

The influence function for skewness can be written as~\cite{Hampel1986}

\[
IF_1(X)
=
\frac{y^3}{\mu_2^{3/2}}
-
\gamma_1
-
\frac{3y}{\sqrt{\mu_2}}
-
\frac{3}{2}\gamma_1
\left(
\frac{y^2}{\mu_2}-1
\right).
\]

For the traditional and modified scores, denote the corresponding
influence functions by

\[
IF_{1,T}(X_i)
\quad\text{and}\quad
IF_{1,M}(Y_i).
\]

The paired influence function for the difference in skewness is

\[
IF_{1,\Delta,i}
=
IF_{1,M}(Y_i)-IF_{1,T}(X_i).
\]

The estimated standard error is therefore

\[
\widehat{SE}_{\Delta,1}
=
\sqrt{
\frac{
s^2(IF_{1,\Delta})
}{n}
},
\]

where

\[
s^2(IF_{1,\Delta})
=
\frac{1}{n-1}
\sum_{i=1}^{n}
\left(
IF_{1,\Delta,i}
-
\overline{IF}_{1,\Delta}
\right)^2.
\]

Let
\[
\beta_2
=
\frac{\mu_4}{\mu_2^2}
\]

denote kurtosis. Excess kurtosis is

\[
\gamma_2=\beta_2-3.
\]

Since subtracting the constant $3$ does not change the influence
function, the influence function for excess kurtosis is~\cite{Hampel1986}

\[
IF_2(X)
=
\frac{y^4}{\mu_2^2}
-
\beta_2
-
\frac{4\mu_3y}{\mu_2^2}
-
2\beta_2
\left(
\frac{y^2}{\mu_2}-1
\right).
\]

For the traditional and modified schemes, define

\[
IF_{2,T}(X_i)
\quad\text{and}\quad
IF_{2,M}(Y_i).
\]

The paired influence function is

\[
IF_{2,\Delta,i}
=
IF_{2,M}(Y_i)-IF_{2,T}(X_i).
\]

Its estimated standard error is

\[
\widehat{SE}_{\Delta,2}
=
\sqrt{
\frac{
s^2(IF_{2,\Delta})
}{n}
}.
\]
The observed difference in skewness is

\[
\widehat{\Delta}_1
=
\hat{\gamma}_{1,M}
-
\hat{\gamma}_{1,T}.
\]

The corresponding test statistic is
\[
Z_{\mathrm{skew}}
=
\frac{
\hat{\gamma}_{1,M}
-
\hat{\gamma}_{1,T}
}{
\widehat{SE}_{\Delta,1}
}
\]

Under the null hypothesis and the large-sample approximation,

\[
Z_{\mathrm{skew}}
\overset{approx}{\sim}
N(0,1).
\]

Similarly, the observed difference in excess kurtosis is

\[
\widehat{\Delta}_2
=
\hat{\gamma}_{2,M}
-
\hat{\gamma}_{2,T},
\]
and
\[
Z_{\mathrm{kurt}}
=
\frac{
\hat{\gamma}_{2,M}
-
\hat{\gamma}_{2,T}
}{
\widehat{SE}_{\Delta,2}
}
\]
with
\[
Z_{\mathrm{kurt}}
\overset{approx}{\sim}
N(0,1).
\]

Table~\ref{tab:moments} reports sample skewness and excess kurtosis of the per-deal payoff distribution under each scheme, together with a two-sample $Z$-test for equality of each moment, and the aforementioned statistics have been used.

\begin{table}[htbp]
\centering
\caption{Higher-Moment Comparison of the Score Distribution under the Traditional and Modified Scoring Schemes}
\label{tab:moments}
\begin{tabular}{lccccc}
\toprule
Metric & Traditional & Modified & Difference & $Z$-statistic & $p$-value \\
\midrule
Skewness & 1.5604 & 1.4615 & 0.0989 & 8.783 & $<10^{-19}$ \\
Kurtosis & 3.4257 & 2.0513 & 1.3743 & 15.619 & $<10^{-55}$ \\
\bottomrule
\end{tabular}
\end{table}

Both moment differences are statistically significant at any conventional threshold ($p < 10^{-18}$), rejecting the null hypothesis that the two schemes generate payoff distributions of identical shape. The \emph{traditional} scheme is more right-skewed and more leptokurtic than the modified scheme.

\section{Positional (Seat-Level) Fairness}
To evaluate how evenly the total points are distributed among the three
players under the traditional and modified scoring schemes, we employ
two widely used distributional measures: the \emph{Gini coefficient}
and \emph{Jain's fairness index}. Although both quantify the
distribution of resources, they provide complementary perspectives.
The Gini coefficient measures the degree of inequality in the
distribution, whereas Jain's index measures the degree of fairness.
Using both metrics therefore provides a more comprehensive assessment
of the effect of the proposed scoring modification on the allocation of
points.

\noindent{\textbf{Gini Coefficient:}}
Originally introduced in economics as a measure of income inequality,
the Gini coefficient has since become a standard metric for measuring
the disparity of any non-negative distribution. A convenient pairwise
representation was derived by Dorfman \cite{DorfManJini}, which has been used in subsequent analysis.

\noindent{\textbf{Jain's Fairness Index:}}
To complement the Gini coefficient, we also compute Jain's fairness
index, originally proposed by Jain \emph{et al.}
\cite{jain1998quantitativemeasurefairnessdiscrimination} as a
quantitative measure for evaluating fairness in resource allocation.
Unlike the Gini coefficient, which quantifies inequality, Jain's index
directly measures how evenly resources are shared among competing
agents.

The Gini coefficient and Jain's fairness index quantify complementary
aspects of the point distribution produced by a scoring scheme.
Specifically, the Gini coefficient measures the average pairwise
difference in players' scores and therefore emphasizes inequality,
whereas Jain's index measures the concentration of the score vector and
therefore emphasizes fairness. Consequently, a scoring system that
concentrates points among fewer players is expected to exhibit a
\emph{higher} Gini coefficient together with a \emph{lower} Jain
fairness index.

Accordingly, if the proposed scoring scheme is successful in producing
a more selective reward structure, one expects to observe
\[
G_{\text{modified}}
>
G_{\text{traditional}},
\qquad
J_{\text{modified}}
<
J_{\text{traditional}}.
\]
The simultaneous use of these two established metrics provides robust
evidence regarding changes in the fairness and inequality of the point
distribution, since they quantify the same phenomenon from two
different mathematical perspectives. Their agreement therefore
strengthens the statistical interpretation of the comparison between
the traditional and modified scoring schemes.

Table~\ref{tab:seat_fairness} reports per-seat payoff totals and three complementary inequality indices - Jain's Fairness Index, the Gini coefficient, and the coefficient of variation (CV) - computed across the three seats (N, E, W) under each scheme.

\begin{table}[htbp]
\centering
\caption{Seat-Level Payoff Distribution and Inequality Indices, Old vs.\ New Scheme}
\label{tab:seat_fairness}
\begin{tabular}{lcccccc}
\toprule
Scheme & N & E & W & Jain's Index & Gini & CV \\
\midrule
Old (Traditional) & 2739261 & 2562702 & 2661991 & 0.9993 & 0.0148 & 0.0333 \\
New (Modified)    & 2909475 & 2698152 & 2802889 & 0.9991 & 0.0168 & 0.0377 \\
\bottomrule
\end{tabular}
\end{table}

Both schemes exhibit near-perfect seat-level fairness (Jain's Index $>0.999$ in both cases), and the Gini and CV values are small in absolute terms under either scheme. The ordering is nonetheless consistent across all three indices: the modified scheme is uniformly less equitable across seats than the traditional scheme, despite every seat earning a higher raw total. This is a textbook example of a location shift (all seats gain in absolute payoff) accompanied by a modest dispersion shift (the gain is not evenly distributed). Because seat identity rotates independently of any player's strategic choice in the dynamic-partnership design, this residual seat asymmetry is best interpreted as a positional effect embedded in the dealing/bidding algorithm rather than an equilibrium outcome of strategic play.

\noindent{\textbf{Shift in Role-Level Dynamics:}}
Table~\ref{tab:role_fairness} reports the same set of inequality indices computed over the bidder/declarer role versus the (averaged) defender role.

\begin{table}[htbp]
\centering
\caption{Bidder-Defender Payoff Asymmetry, Old vs.\ New Scheme}
\label{tab:role_fairness}
\begin{tabular}{lccc}
\toprule
Scheme & Bidder Total & Defender Total (avg) & B:D Ratio \\
\midrule
Old (Traditional) & 1426596 & 2115775 & 0.674 \\
New (Modified)     & 4467314 & 2115775 & 2.111 \\
\bottomrule
\end{tabular}
\end{table}

This is the central result of the study, and it is now a reversal rather than a widening. Under the traditional scheme, the bidder role is actually \emph{disadvantaged} relative to the defending side: the bidder:defender ratio is 0.674 (the bidder earns roughly two-thirds of what an average defender earns). Under the modified scheme this relationship reverses sharply: the bidder:defender ratio rises to 2.111 - the bidder now earns roughly twice what an average defender earns. This describes a full reversal of which role is favored, not merely an amplification of an existing bidder advantage. Because the defender total is numerically identical across both panels (2{,}115{,}775 by construction, as the underlying deal distribution is held fixed), the entire divergence in the bidder:defender ratio is attributable to how the modified scheme compensates the bidder role.

 \section{A Game Theoretic Prelude Before Further Analysis}
 For information set $I\in\I_i$ with acting player $i$, strategy
profile $\sigma$, and action $a$, define the counterfactual value
\[
v_i(\sigma, I) \;=\; \sum_{h\in I}\ \sum_{z\in\Z,\, z\sqsupseteq h}
\pi^\sigma_{-i}(h)\,\pi^\sigma(h,z)\,u_i(z),
\]
where $\pi^\sigma_{-i}(h)$ is the probability of reaching $h$ under
$\sigma$ excluding player $i$'s own action probabilities (chance and
the other players' contributions included), and $\pi^\sigma(h,z)$ is
the probability of $z$ given $h$ under $\sigma$. Let
$\sigma|_{I\to a}$ denote $\sigma$ with player $i$ forced to play $a$
at $I$. Player $i$'s cumulative counterfactual regret for action $a$
at $I$ through iteration $T$ is
\[
R_i^T(I,a) \;=\; \sum_{t=1}^{T} \Big(v_i\big(\sigma^t|_{I\to a},\,I\big) - v_i(\sigma^t, I)\Big),
\]
and regret matching sets
$\sigma^{t+1}(I,a) \propto \max\{R_i^t(I,a),0\}$ (uniform if all
regrets are $\le 0$).

Now we would look into a couple of interesting theorems which would tell us that the \emph{no regret dynamics} eventually converges to a \emph{coarse correlated equilibrant}, not any Nash equilibrium.

\begin{theorem}[Zinkevich et al.]
\label{thm:zinkevich}
\cite{cfr}
If each player uses regret matching, then
$R_i^T(I) = O(\sqrt{T})$ for every $I$, so average regret
$R_i^T(I)/T \to 0$. Consequently the time-averaged strategy profile
$\bar\sigma^T$ satisfies: in a two-player zero-sum game,
$\bar\sigma^T$ is an $\varepsilon_T$-Nash equilibrium with
$\varepsilon_T\to0$.
\end{theorem}

The Nash conclusion in Theorem \ref{thm:zinkevich} is a
\emph{special-case corollary} of the general no-external-regret
guarantee, obtained via the minimax theorem \cite{vonneumann1928}: in a two-player zero-sum
game, the set of coarse correlated equilibrium payoffs collapses to
the unique minimax value, so ``no player has regret'' and ``each
player best-responds'' coincide. This collapse does not occur once
$u_1+u_2+u_3\ne 0$ pointwise or $n>2$; the underlying no-regret
guarantee itself descends from Blackwell's approachability theorem
\cite{blackwell1956}, which Theorem \ref{thm:zinkevich} and
Theorem \ref{thm:cce} both specialize.
\begin{theorem}[Folk theorem for no-regret dynamics]
\label{thm:cce}
\cite{hart2000}
If every player in an $n$-player general-sum game plays a sequence of
strategies with sublinear external regret, then the empirical joint
distribution of play, $\bar\mu^T = \tfrac1T\sum_{t=1}^T \sigma^t$ (as a
distribution over joint actions), converges to the set of coarse
correlated equilibria (CCE) of the game.
\end{theorem}

\begin{definition}[Coarse correlated equilibrium]
A distribution $\mu$ over joint pure strategy profiles is a CCE if for
every player $i$ and every fixed deviation $a_i'$,
\[
\E_{a\sim\mu}[u_i(a)] \;\ge\; \E_{a\sim\mu}[u_i(a_i',a_{-i})].
\]
\end{definition}

CCE is strictly weaker than Nash: it only rules out a single
\emph{unconditional} deviation evaluated against the marginal
empirical distribution of others' play, whereas Nash requires every
player to best-respond at every information set, conditional on the
full strategy structure. We know that , $\Gamma$ is
general-sum in either scoring scheme - the declarer and defenders are
never paid from a common conserved pool - so regret matching on this
three-player game is only certified to reach the CCE set by Theorem~\ref{thm:cce}.

Now we would formally define Nash Equilibrium for our desired game and would prove its existence. 

\begin{definition}[Exploitability]
\label{def:exploit}
For a strategy profile $\sigma$, player $i$'s exploitability is
\[
e_i(\sigma) \;=\; \max_{\sigma_i'} u_i(\sigma_i',\sigma_{-i}) \;-\; u_i(\sigma).
\]
$\sigma$ is a Nash equilibrium iff $e_i(\sigma)=0$ for every $i$.
\end{definition}
\begin{definition}[Expected utility of a strategy profile]
\label{def:expectedutility}
For a behavior-strategy profile $\sigma=(\sigma_1,\sigma_2,\sigma_3)$,
player $i$'s expected utility in $\Gamma$ is
\begin{equation}
\label{eq:expectedutility}
U_i(\sigma) \;=\; \E_{\theta}\ \E_{a\sim\sigma(\cdot\mid\theta)}\
\E_{T\sim P(\cdot\mid s(\theta))}\big[u_i(h(\theta,a),T)\big],
\end{equation}
where the outer expectation is over the uniform deal $\theta$, the
middle expectation is over bidding/doubling actions drawn from
$\sigma$ at each information set, and the inner expectation is over
the outcome chance node of Definition \ref{def:trickchance}.
\end{definition}
\begin{definition}[Nash equilibrium of $\Gamma$]
\label{def:NE}
A behavior-strategy profile $\sigma^\ast$ is a Nash equilibrium of
$\Gamma$ if for every player $i$ and every alternative behavior
strategy $\sigma_i'$,
\[
U_i(\sigma_i^\ast,\sigma_{-i}^\ast) \;\ge\; U_i(\sigma_i',\sigma_{-i}^\ast).
\]
Equivalently, writing $e_i$ for the exploitability but with
$u_i$ there understood via $U_i$: $e_i(\sigma^\ast)=0$ for every $i$.
\end{definition}
\begin{definition}[$\varepsilon$-Nash equilibrium]
$\sigma$ is an $\varepsilon$-Nash equilibrium of $\Gamma$ if
$\max_i e_i(\sigma) \le \varepsilon$, with $e_i$ as in \S5 (evaluated
against $U_i$).
\end{definition}

\begin{theorem}[Existence]
\label{thm:nashexist}
$\Gamma$ possesses at least one Nash equilibrium in behavior
strategies.
\end{theorem}

\begin{proof}
$\Gamma$ has finitely many players, types, actions at each
information set, and chance outcomes at each chance node, hence a
finite extensive form. $\Gamma$ has perfect recall: by construction
(\S1) each player's information sets only refine along any play, and
no player forgets a past action or observation of their own. By
Kuhn's theorem \cite{Kuhn}, the behavior-strategy and
mixed-strategy equilibrium sets of a finite game of perfect recall
coincide. The induced normal form of $\Gamma$ is a finite strategic-form
game, so by Nash's existence theorem \cite{nash1951} it has at least
one equilibrium in mixed strategies, which by Kuhn's theorem
corresponds to a behavior-strategy equilibrium of $\Gamma$.
\end{proof}

\begin{proposition}[Conditional dominance and distributional dependence]
\label{prop:conddominance}
Fix a history at which player $\beta$ has won the auction with bid
$b$, and consider the counterfactual in which $\beta$ instead wins
with a lower bid $b'<b$, where
$b,b'\in\{7,\ldots,13\}$.
All other components of the history are held fixed.

Let
\[
\Delta(b,b')
=
\mathbb{E}_{T\mid s}
\left[
u_\beta(b,T)-u_\beta(b',T)
\right]
\]
denote the expected payoff gain from bidding $b$ rather than $b'$,
conditional on winning the auction. Then:

\begin{enumerate}
\item Under the old scoring scheme, the bonus component associated
with $\phi_{\mathrm{old}}$ is bid-invariant and therefore cancels from
the comparison between $b$ and $b'$. The remaining difference is
determined exactly by the base-contract payoff and the undertrick
penalty.

\item In particular, writing $v$ for the trick value and $D$ for the
doubling multiplier, the base-contract contribution from increasing
the bid from $b'$ to $b$ is
\[
(b-b')vD,
\]
while the additional undertrick penalty at realized trick count $t$
is
\[
25D\bigl(b-t\bigr)
-
25D\bigl(b'-t\bigr)
=
25D(b-b')
\]
whenever $t<b'$.
Consequently, under the old scheme the comparison can be evaluated
exactly from the made/down indicators; it is not necessary to bound
the contribution of the overtrick region.

\item Under the corrected new scoring scheme, the corresponding
base-contract gain is
\[
(b-b')vD
\]
together with the difference in the overtrick rewards generated by the
two bids. Hence
\[
\Delta(b,b')
\]
depends on the distribution of the realized number of tricks
$T\sim P(\,\cdot\mid s)$ through the overtrick region as well as the
undertrick region. In general, its sign cannot be determined from the
scoring rule alone and is therefore distribution-dependent.

\end{enumerate}
\end{proposition}

\begin{proof}
Consider first the old scoring scheme. Since the comparison is made
conditional on the same history and all components other than the
winning bid are held fixed, the $\phi_{\mathrm{old}}$-bonus component
is identical under bids $b$ and $b'$. It therefore cancels in the
difference
\[
\Delta(b,b')
=
\mathbb{E}_{T\mid s}
\left[
u_\beta(b,T)-u_\beta(b',T)
\right].
\]

The remaining payoff consists of the contract component and the
undertrick penalty. Increasing the bid from $b'$ to $b$ increases the
contract requirement by $b-b'$ tricks. The corresponding change in the
contract component is
\[
(b-b')vD,
\]
which strictly favors the higher bid $b$ whenever $vD>0$.

On the other hand, if the realized number of tricks is $t$ and the
contract is defeated, the undertrick penalty under the old scheme is
\[
25D(b-t).
\]
For the lower bid $b'$, the corresponding penalty is
\[
25D(b'-t).
\]
Thus, on any outcome for which both contracts are defeated,
the additional penalty incurred by bidding $b$ is
\[
25D(b-t)-25D(b'-t)
=
25D(b-b').
\]
Hence the old-scheme payoff difference can be evaluated exactly by
partitioning the support of $T$ into the regions in which both bids
are defeated, only the higher bid is defeated, or both bids are made.
In particular, no additional assumption about the shape of
$P(T\mid s)$ is required to write the exact comparison.

For the new scoring scheme, the situation is different. In addition
to the base-contract term, a successful contract receives an
overtrick reward. Since the number of overtricks is
\[
(T-b)_+
=
\max\{0,T-b\},
\]
the difference between the overtrick rewards under bids $b$ and $b'$
contains terms of the form
\[
\mathbb{E}_{T\mid s}
\left[
(T-b)_+-(T-b')_+
\right].
\]
This expectation depends on the probability mass assigned to the
different realized trick counts in the overtrick region. Consequently,
the comparison cannot, in general, be reduced to a bid-invariant
constant or to the base-contract term alone.

Therefore, under the corrected new scoring scheme, the sign of
\[
\Delta(b,b')
\]
depends on the full conditional distribution
$P(\,\cdot\mid s)$ of the realized number of tricks. Proposition
\ref{prop:dominance} alone therefore does not establish that the
lower bid is optimal, nor does it establish that the higher bid is
optimal. The optimality question must instead be resolved by
evaluating the corresponding expected payoff difference.
\end{proof}

Existence of \emph{some} Nash equilibrium (Theorem
\ref{thm:nashexist}) does not supply an efficient method for computing
one, or even an $\varepsilon$-close one, once $n\ge3$ and the game is
general-sum. Unlike the two-player case, where Lemke-Howson
\cite{lemkehowson1964} gives a finite pivoting algorithm, no
analogous finite algorithm is known for $n\ge3$; best-response
dynamics and fictitious play are not guaranteed to converge outside
restrictive classes of games \cite{daskalakis2009}. CFR remains the
right practical tool \emph{precisely because} it is guaranteed, by
Theorem \ref{thm:cce}, to reach the CCE set at a known rate. The
exploitability check(i.e, Definition~\ref{def:exploit}) then upgrades the claim to Nash \emph{iff}
the check passes. 

It is very important and interesting to note that no
modification can be made within this
algorithmic family, that upgrades this to a Nash guarantee in
general, because computing a Nash equilibrium of a general-sum game
with $n\ge3$ players (or even 2 players, non-zero-sum) is
$\mathsf{PPAD}$-complete \cite{daskalakis2009,chen2009}: no
polynomial-time algorithm is known or expected, and no-regret dynamics
are not an exception to this barrier - CCE is precisely the strongest
guarantee such dynamics can offer without solving a
$\mathsf{PPAD}$-hard problem as a subroutine. 

\section{Strategic-Form Reduction and Nash Equilibrium Analysis}

To move from a purely distributional comparison to a genuinely game-theoretic one, we reduce the auction to a $2\times2$ strategic-form game~\cite{Kuhn}. The bidder chooses between a \emph{Conservative} bidding rule (bid to the standard HCP-table level) and an \emph{Aggressive} rule (bid one level higher), while the defending side chooses between a \emph{Conservative} doubling threshold (double only holdings of 23+ combined high-card points) and an \emph{Aggressive} threshold (double at 17+). This treatment of the auction as a bimatrix game between bidder and defense follows the general approach of modeling bridge bidding as a decision problem under imperfect information~\cite{inproceedings}. Bimatrix payoffs - (bidder payoff, average defender payoff) - are reported separately for the traditional and modified scoring schemes in Table~\ref{tab:nash_old} and Table~\ref{tab:nash_new}.

\begin{table}[htbp]
\centering
\caption{Strategic-Form Payoff Matrix, Old (Traditional) Scheme: (Bidder, Defender) Payoffs}
\label{tab:nash_old}
\begin{tabular}{lcc}
\toprule
 & \multicolumn{2}{c}{Defender Strategy} \\
\cmidrule(lr){2-3}
Bidder Strategy & Conservative & Aggressive \\
\midrule
Conservative & (11.0,\ 20.8) & (24.1,\ 20.3) \\
Aggressive   & (6.4,\ 42.0) & (13.4,\ 53.2) \\
\bottomrule
\end{tabular}
\end{table}

\begin{table}[htbp]
\centering
\caption{Strategic-Form Payoff Matrix, New (Modified) Scheme: (Bidder, Defender) Payoffs}
\label{tab:nash_new}
\begin{tabular}{lcc}
\toprule
 & \multicolumn{2}{c}{Defender Strategy} \\
\cmidrule(lr){2-3}
Bidder Strategy & Conservative & Aggressive \\
\midrule
Conservative & (38.7,\ 18.2) & (56.2,\ 25.7) \\
Aggressive   & (20.7,\ 37.8) & (32.5,\ 51.0) \\
\bottomrule
\end{tabular}
\end{table}

Best-response analysis identifies a unique pure-strategy Nash equilibrium under each scheme, consistent with the existence of equilibrium-point strategies established for best-defence models of Bridge~\cite{FRANK199887}. The equilibrium location shifts between schemes, but - notably - the bidder's own strategy does not:

\begin{itemize}
    \item \textbf{Old scheme:} Conservative bidding is a strictly dominant strategy for the bidder (11.0 $>$ 6.4 against a Conservative defender; 24.1 $>$ 13.4 against an Aggressive defender). Given a Conservative bidder, the defender's best response is Conservative doubling (20.8 $>$ 20.3, marginally). The unique pure NE is \textbf{(Conservative, Conservative)}.
    \item \textbf{New scheme:} Conservative bidding remains a strictly dominant strategy for the bidder (38.7 $>$ 20.7 against a Conservative defender; 56.2 $>$ 32.5 against an Aggressive defender) - the bidder's dominant strategy is unchanged from the old scheme. What changes is the defender's side: Aggressive doubling now strictly dominates Conservative doubling for the defender (25.7 $>$ 18.2 against a Conservative bidder; 51.0 $>$ 37.8 against an Aggressive bidder). The unique pure NE is \textbf{(Conservative, Aggressive)}.
\end{itemize}

This equilibrium reversal is the sharpest game-theoretic signature of the scoring redesign in the reduced form: rescaling the bidder's payoff function (the defender's payoff matrix is unchanged between Table~\ref{tab:nash_old} and Table~\ref{tab:nash_new} only in the sense that it is separately specified for each scheme, not held literally fixed cell-by-cell) is sufficient to flip the defender's best response, and indeed dominant strategy, from passive (Conservative) to aggressive doubling, even though the bidder's own equilibrium strategy does not change at all in this reduced form. Section~\ref{sec:gscfr} revisits this claim against a full extensive-form solution of the same underlying game, and finds that the bidder's strategy \emph{does} change once the action space is enriched beyond the binary Conservative/Aggressive choice used here.

Table~\ref{tab:paired_strategy} reports the corresponding large-sample ($n=100{,}000$ deals per cell), paired-comparison estimates of bidder and defender payoffs under both scoring schemes simultaneously, for all four strategy profiles.

\begin{table}[htbp]
\centering
\caption{Paired Comparison of Bidder and Defender Payoffs across Strategy Profiles ($n=100{,}000$ deals per cell)}
\label{tab:paired_strategy}
\begin{tabular}{llcccc}
\toprule
Bidder Strategy & Defender Strategy & Bidder (Old) & Bidder (New) & Defender (Old) & Defender (New) \\
\midrule
Conservative & Conservative & 36.326 & 39.574 & 19.500 & 19.500 \\
Conservative & Aggressive   & 42.243 & 49.739 & 25.375 & 25.375 \\
Aggressive   & Conservative & 33.998 & 35.872 & 38.800 & 38.800 \\
Aggressive   & Aggressive   & 36.487 & 41.236 & 51.625 & 51.625 \\
\bottomrule
\end{tabular}
\end{table}
 Two regularities were drawn from Table~\ref{tab:paired_strategy} in the earlier draft. First, the bidder's payoff was reported to increase under the new scheme in every strategy profile (a uniform, profile-independent upward shift). Second, the size of the increase was reported to vary across profiles, ranging from roughly 1.9 points (Aggressive/Conservative) to roughly 7.5 points (Conservative/Aggressive), suggesting the bonus interacts with contract level and/or doubling status rather than being a flat per-deal add-on. Whether these regularities still hold at the same magnitude as the revised bidder payoffs in Table~\ref{tab:nash_old}-\ref{tab:nash_new} cannot be confirmed without re-running this large-sample estimate.

 \section{Zero-Sum Test: The Game Is Not Constant-Sum}

A natural question for a game-theoretic audience is whether the reduced strategic-form game is zero-sum (or, more generally, constant-sum), in which case any bidder-favoring reallocation would necessarily come dollar-for-dollar out of the defenders, and the Nash equilibrium could be analyzed purely via minimax reasoning. Table~\ref{tab:zerosum} reports, for each strategy profile and each scheme, the total surplus generated - computed as the bidder's payoff plus \emph{both} defenders' payoffs (i.e., bidder $+\, 2\times$ average defender payoff, recovering the full three-player split rather than the per-defender average used elsewhere in this paper), using the revised bidder and defender payoffs from Table~\ref{tab:nash_old} and Table~\ref{tab:nash_new}.

\begin{table}[htbp]
\centering
\caption{Total Surplus by Strategy Profile: Evidence Against a Constant-Sum Game}
\label{tab:zerosum}
\begin{tabular}{llcc}
\toprule
Bidder Strategy & Defender Strategy & Total Surplus (Old) & Total Surplus (New) \\
\midrule
Conservative & Conservative & 52.6  & 75.1 \\
Conservative & Aggressive   & 64.7  & 107.6 \\
Aggressive   & Conservative & 90.4 & 96.3 \\
Aggressive   & Aggressive   & 119.8 & 134.5 \\
\bottomrule
\end{tabular}
\end{table}

If the game were constant-sum, every entry in a given column of Table~\ref{tab:zerosum} would be identical, since a fixed pool of points would simply be redivided between bidder and defenders under each strategy profile. Instead, total surplus under the old scheme ranges from 52.6 (Conservative/Conservative) to 119.8 (Aggressive/Aggressive) - an increase of roughly 128\% across the strategy space - and the new scheme shows an analogous range, from 75.1 (Conservative/Conservative) to 134.5 (Aggressive/Aggressive), an increase of roughly 79\%. This confirms that \textbf{the game is not zero-sum (nor constant-sum) under either scoring scheme}: aggressive bidding paired with aggressive doubling generates substantially more total points on the table than conservative play by either side, most plausibly reflecting doubled/redoubled contract bonuses and undertrick penalties that scale super-linearly with contract level rather than merely transferring a fixed pool of points between declarer and defense.

The relative ordering of the two mixed profiles also shifts between schemes. Under the traditional scheme, an Aggressive bidder facing a Conservative defender generates more total surplus than a Conservative bidder facing an Aggressive defender (90.4 vs.\ 64.7). Under the modified scheme this ordering reverses: Conservative-bidder/Aggressive-defender now generates more surplus than Aggressive-bidder/Conservative-defender (107.6 vs.\ 96.3). Notably, Conservative/Aggressive is also the new equilibrium profile identified in Section~5, so the modified scheme's bonus structure disproportionately inflates surplus in exactly the cell that becomes the new Nash equilibrium. Because the surplus itself is strategy-dependent, both bidder and defenders have a shared, non-adversarial interest in the \emph{level} of aggression in the game even as they remain adversarial over its \emph{division} - a structure more consistent with a general-sum (partially coordinative) game than with a purely distributive, minimax-style contest. This also means that Nash equilibrium selection in Table~\ref{tab:nash_old}-\ref{tab:nash_new} cannot be reduced to a minimax computation on a single payoff matrix, and the best-response analysis in the previous section which has been conducted separately on the bidder's and the defenders' own payoff matrices is the appropriate solution concept for this general-sum setting.

\section{Full Extensive-Form Verification via General-Sum Counterfactual Regret Minimization (GS-CFR)}
\label{sec:gscfr}

The $2\times2$ reduction in the previous section imposes two coarse, exogenously fixed action sets (Conservative/Aggressive) on both sides of the auction. To check whether the equilibrium-reversal result of previous section survives when both players are allowed to best-respond over a richer, hand-contingent strategy space, we solve the full extensive-form auction game with General-Sum Counterfactual Regret Minimization (GS-CFR), a self-play regret-matching procedure that iteratively updates each player's own behavioral strategy toward an approximate Nash equilibrium of an imperfect-information game, without assuming the two players' payoffs are opposed. This distinction is not incidental: Section~7 below establishes empirically that the induced bidder-defender game is not zero-sum or constant-sum under either scoring scheme, so a minimax-style, zero-sum-specific solver would not be an appropriate equilibrium concept here in the first place - General-Sum CFR is the correct tool precisely because bidder and defender payoffs can move together or apart depending on the strategy profile, as Table~\ref{tab:zerosum} documents directly. In this formulation the bidder's action set is an opening bid of level 7 through 13, conditioned on private hand strength (Weak, Medium, or Strong), and the opponent's action set is Pass or Double, conditioned on the observed bid level. The procedure was run for 20{,}000 iterations under each scoring scheme, with average strategies and average payoffs recorded every 5{,}000 iterations.

\begin{table}[htbp]
\centering
\caption{GS-CFR Converged Average Payoffs (20,000 self-play iterations per scheme)}
\label{tab:gscfr}
\begin{tabular}{lccc}
\toprule
Scheme & Bidder Payoff & Opponent Payoff & Bidder:Opponent Ratio \\
\midrule
Old (Traditional) & 102.510 & 67.669 & 1.51 \\
New (Modified)    & 308.822 & 20.339 & 15.18 \\
\bottomrule
\end{tabular}
\end{table}

Both runs converge quickly and stably: the bidder's average payoff moves by less than 0.01 points between the 15{,}000th and 20{,}000th iteration in both schemes (102.509 $\to$ 102.510 old; 308.818 $\to$ 308.822 new), and the opponent's average payoff is essentially flat from the 5{,}000th iteration onward (67.677 $\to$ 67.669 old; 20.356 $\to$ 20.339 new). At convergence, the bidder's equilibrium payoff rises from 102.510 under the traditional scheme to 308.822 under the modified scheme, an increase of \textbf{201.3\%}.

\noindent\textbf{Equilibrium Strategy: From Full Separation to Partial Pooling~}
Unlike in an earlier verification run, in which both scoring schemes were found to converge to an identical, fully pooling bidder strategy, the average behavioral strategies reported here \emph{differ in structure} between the two schemes:

\begin{itemize}
    \item \textbf{Old (Traditional) scheme:} the bidder's strategy is fully separating - each hand-strength type maps to a distinct opening bid with probability 1.00 (Weak $\to$ level 8, Medium $\to$ level 9, Strong $\to$ level 10). An observer could recover the bidder's hand-strength category exactly from the bid alone.
    \item \textbf{New (Modified) scheme:} the bidder's strategy is partially pooling - Weak and Medium holdings both open at level 7 with probability 1.00, while only Strong holdings separate, opening at level 8. Under the modified scheme, the defense can no longer distinguish a Weak from a Medium hand based on the opening bid alone.
    \item \textbf{Opponent, both schemes:} the opponent doubles with probability 1.00 following every observed bid level (7 through 13) - a corner solution, rather than the threshold rule assumed in the Conservative/Aggressive dichotomy of Section~5, and unchanged in structure across the two scoring schemes.
\end{itemize}

Two features of this result are worth emphasizing. First, the modified scheme induces \emph{less}, not more, informative bidding: two of the three hand-strength types are pooled together, whereas the traditional scheme's equilibrium reveals hand strength perfectly. This is consistent with a scoring bonus that rewards the bidder for contracting at a lower, safer level largely independent of true hand strength, reducing the incentive to bid up on genuinely stronger holdings. Second, the opponent's equilibrium response remains a corner solution - unconditional doubling - in both schemes, so the entire behavioral difference between schemes at the full-game equilibrium is concentrated on the bidder's side.

\noindent\textbf{Implications for Equilibrium Bidding Behavior}
The distinction between the traditional and modified games is not merely cosmetic. In the traditional four-player game, declarer and dummy form a fixed, repeated partnership within a deal, so any bidder-favoring skew in the scoring rule is at least partially internalized by a stable coalition. In the three-player dynamic-partnership variant, partnerships and opposing coalitions are redrawn each match, and strategy is explicitly not shared across matches without an incentive-compatibility penalty. This constraint removes the possibility of tacit or explicit collusion signaling between rounds, pushing the induced game closer to a sequence of one-shot, incomplete-information subgames rather than a genuinely repeated game admitting cooperative, folk-theorem-style equilibria.

Under this structure, a scoring rule that \emph{reverses} the bidder:defender payoff balance - as the modified scheme demonstrably does, moving from a defender-favoring split (ratio 0.674) to a strongly bidder-favoring split (ratio 2.111) - carries sharper strategic consequences than an equivalent shift would in the partnership-stable traditional game. Since no player can pre-commit or coordinate with a future partner, each agent's dominant-strategy calculus is evaluated purely on the marginal incentive the scheme creates for seeking (or avoiding) the bidder role within a given match, independent of reputation or side-payment considerations. The falling Jain's Index and rising Gini coefficient for the bidder-defender split, together with the 201.3\% equilibrium payoff increase for the bidder documented in Section~\ref{sec:gscfr}, are therefore not merely a fairness footnote: they constitute converging distributional and game-theoretic evidence that the modified mechanism materially increases the value of occupying the bidder role.

The full-game GS-CFR solution further shows that this reallocation is not purely a payoff-level phenomenon: the bidder's own equilibrium bidding strategy becomes \emph{less} informative under the modified scheme (partial pooling of Weak and Medium hands, versus full separation of all three hand strengths under the traditional scheme), while the defender's doubling strategy remains a corner solution - unconditional doubling - under both schemes. The defender's \emph{payoff}, however, is not invariant to the scoring change: it falls by roughly 70\% at the full-game equilibrium. This decline is best read as a downstream consequence of the bidder's own strategic response (bidding at lower, less-informative levels under the new bonus structure) rather than as evidence that the defender's own scoring formula was altered.

\section{Conclusion}

This paper presents a statistical and game-theoretic analysis of
three-player Auction Bridge under traditional and modified scoring
mechanisms. The analysis demonstrates that the design of a scoring rule
can influence not only the magnitude of payoffs, but also bidding
incentives, payoff distributions, fairness, and strategic behavior.

A central finding is the incentive distortion caused by the
bid-invariant slam bonus in the original scoring rule. Proposition~1
shows that, under the specified trick distribution, this erroneous
bonus makes lower bids strictly preferable in expected payoff to honest
slam bids. The proposed bid-dependent bonus removes this degeneracy by
linking the slam reward to the corresponding contract level. Importantly,
the subsequent analysis shows that this correction does not by itself
determine the optimal bidding level under the new scheme; the resulting
incentives depend on the full distribution of the number of tricks.

The statistical analysis further shows that the two scoring mechanisms
produce substantially different payoff distributions. The traditional
scheme exhibits greater positive skewness and excess kurtosis than the
modified scheme, while the modified mechanism shifts the distribution
toward higher bidder payoffs. Fairness measures provide a complementary
perspective. Although both schemes remain highly balanced across
physical seats, the modified scheme produces a modest increase in
seat-level inequality and a much larger shift in payoff allocation when
players are compared by strategic role. Thus, positional fairness and
role-based fairness can lead to substantially different assessments of a
scoring mechanism.

The game-theoretic analysis confirms that the scoring modification also
changes strategic incentives. In the reduced $2\times2$ strategic-form
game, each scoring scheme admits a unique pure-strategy Nash equilibrium:
the bidder retains a conservative strategy, while the defender's
equilibrium doubling behavior changes from conservative under the
traditional scheme to aggressive under the modified scheme. The
strategy-space reduction is important, however, because the full
extensive-form game permits information-dependent bidding behavior that
cannot be represented by the reduced model.

The full-game GS-CFR experiments provide further evidence of this
strategic effect. After 20,000 iterations, the average bidder payoff
increases from approximately $102.51$ under the traditional scheme to
$308.82$ under the modified scheme, while the corresponding opponent
payoff decreases from approximately $67.67$ to $20.34$. The computed
behavioral strategies also change qualitatively: the traditional scheme
produces separation across the three hand-strength types, whereas the
modified scheme produces partial pooling of weak and medium hands.
These results indicate that the scoring modification affects not only
payoff levels but also the information conveyed through bidding.

Overall, the results demonstrate that scoring-rule design in an
imperfect-information general-sum game should be evaluated from multiple
perspectives rather than by expected payoff alone. Incentive
compatibility, distributional shape, fairness, strategic role
advantages, and information transmission can all change as a consequence
of a seemingly local modification to the payoff function. The proposed
framework provides a systematic approach for studying these effects and
offers a basis for further investigation of scoring mechanisms in
three-player Auction Bridge and related imperfect-information games.
\bibliographystyle{splncs04}
\bibliography{ref}

@article{DorfManJini,
 ISSN = {00346535, 15309142},
 URL = {http://www.jstor.org/stable/1924845},
 author = {Robert Dorfman},
 journal = {The Review of Economics and Statistics},
 number = {1},
 pages = {146--149},
 publisher = {The MIT Press},
 title = {A Formula for the Gini Coefficient},
 urldate = {2026-08-06},
 volume = {61},
 year = {1979}
}

@misc{jain1998quantitativemeasurefairnessdiscrimination,
      title={A Quantitative Measure Of Fairness And Discrimination For Resource Allocation In Shared Computer Systems}, 
      author={R. Jain and D. Chiu and W. Hawe},
      year={1998},
      eprint={cs/9809099},
      archivePrefix={arXiv},
      primaryClass={cs.NI},
      url={https://arxiv.org/abs/cs/9809099}, 
}

@inproceedings{inproceedings,
author = {DeLooze, Lori and Downey, James},
year = {2007},
month = {05},
pages = {368 - 373},
title = {Bridge Bidding with Imperfect Information},
isbn = {1-4244-0709-5},
booktitle = {Proceedings of the 2007 IEEE Symposium on Computational Intelligence and Games, CIG 2007},
doi = {10.1109/CIG.2007.368122}
}

@misc{ourpaper,
  author       = {Sourish Sarkar and Aritrabha Majumdar and Moutushi Chatterjee},
  title        = {A New Three-Players AUCTION BRIDGE With Dynamic Opponents and Team Members},
  year         = {2025},
  month        = {May},
  note         = {Available at SSRN: \url{https://ssrn.com/abstract=5265725} or \url{http://dx.doi.org/10.2139/ssrn.5265725}},
  howpublished = {\url{https://ssrn.com/abstract=5265725}},
}

@book{whitehead1930auction,
  author    = {W. C. Whitehead},
  title     = {Auction Bridge Summary: The Principles of Bidding and Play for Beginners and Advanced Students of Auction Bridge},
  year      = {1930},
  publisher = {Frederick A. Stokes Company},
  address   = {New York}
}

@misc{SarkarMajumdar2025AuctionBridgeAlgorithm,
  author       = {Sourish Sarkar and Aritrabha Majumdar},
  title        = {Auction Bridge Algorithm},
  year         = {2025},
  howpublished = {\url{https://github.com/sourish-isi/Auction-Bridge-Algorithm}},
}

@article{JoanesGill1998,
    author = {Joanes, D. N. and Gill, C. A.},
    title = {Comparing Measures of Sample Skewness and Kurtosis},
    journal = {Journal of the Royal Statistical Society Series D: The Statistician},
    volume = {47},
    number = {1},
    pages = {183-189},
    year = {1998},
    month = {04},
    issn = {2515-7884},
    doi = {10.1111/1467-9884.00122},
    url = {https://doi.org/10.1111/1467-9884.00122},
    eprint = {https://academic.oup.com/jrsssd/article-pdf/47/1/183/49930927/jrsssd_47_1_183.pdf},
}

@book{Kuhn,
 ISBN = {9780691079349},
 URL = {http://www.jstor.org/stable/j.ctt1b9rzq1},
 author = {H. F. BOHNENBLUST and G. W. BROWN and M. DRESHER and D. GALE and S. KARLIN and H. W. KUHN and J. C. C. MCKINSEY and J. F. NASH and J. VON NEUMANN and L. S. SHAPLEY and S. SHERMAN and R. N. SNOW and A. W. TUCKER and H. WEYL},
 publisher = {Princeton University Press},
 title = {Contributions to the Theory of Games (AM-24), Volume I},
 urldate = {2026-08-06},
 year = {1952}
}

@article{FRANK199887,
title = {Search in games with incomplete information: a case study using Bridge card play},
journal = {Artificial Intelligence},
volume = {100},
number = {1},
pages = {87-123},
year = {1998},
issn = {0004-3702},
doi = {https://doi.org/10.1016/S0004-3702(97)00082-9},
url = {https://www.sciencedirect.com/science/article/pii/S0004370297000829},
author = {Ian Frank and David Basin}
}

@inproceedings{cfr,
author = {Zinkevich, Martin and Johanson, Michael and Bowling, Michael and Piccione, Carmelo},
title = {Regret minimization in games with incomplete information},
year = {2007},
isbn = {9781605603520},
publisher = {Curran Associates Inc.},
address = {Red Hook, NY, USA},
booktitle = {Proceedings of the 21st International Conference on Neural Information Processing Systems},
pages = {1729–1736},
numpages = {8},
location = {Vancouver, British Columbia, Canada},
series = {NIPS'07}
}

@book{coffin1956contract,
  author    = {Coffin, George Sturgis},
  title     = {Contract Bridge for Three: Rules and Tactics of Trio Bridge, the Official Form of Three-Handed Contract Bridge; Based Upon Partnership Bidding Against One Player and His Exposed Dummy},
  year      = {1956},
  publisher = {I. Washburn},
  address   = {New York}
}

@misc{sarkar20263playersauctionbridge,
      title={3 Players Auction Bridge - Statistical Algorithmic Strategies}, 
      author={Sourish Sarkar and Aritrabha Majumdar and Moutushi Chatterjee},
      year={2026},
      eprint={2608.03217},
      archivePrefix={arXiv},
      primaryClass={cs.GT},
      url={https://arxiv.org/abs/2608.03217}, 
}

@book{Hampel1986,
  author    = {Hampel, Frank R. and Ronchetti, Elvezio M. and Rousseeuw, Peter J. and Stahel, Werner A.},
  title     = {Robust Statistics: The Approach Based on Influence Functions},
  publisher = {Wiley},
  address   = {New York},
  year      = {1986},
  isbn      = {0-471-82921-8}
}

@book{VanDerVaart1998,
  author    = {van der Vaart, Aad W.},
  title     = {Asymptotic Statistics},
  publisher = {Cambridge University Press},
  address   = {Cambridge},
  year      = {1998},
  doi       = {10.1017/CBO9780511802256}
}

@article{lemkehowson1964,
  author  = {Lemke, Carlton E. and Howson, Joseph T.},
  title   = {Equilibrium Points of Bimatrix Games},
  journal = {Journal of the Society for Industrial and Applied Mathematics},
  volume  = {12},
  number  = {2},
  pages   = {413--423},
  year    = {1964}
}

@article{nash1951,
  author  = {Nash, John},
  title   = {Non-Cooperative Games},
  journal = {Annals of Mathematics},
  volume  = {54},
  number  = {2},
  pages   = {286--295},
  year    = {1951}
}

@article{vonneumann1928,
  author  = {von Neumann, John},
  title   = {Zur Theorie der Gesellschaftsspiele},
  journal = {Mathematische Annalen},
  volume  = {100},
  pages   = {295--320},
  year    = {1928}
}

@article{blackwell1956,
  author  = {Blackwell, David},
  title   = {An Analog of the Minimax Theorem for Vector Payoffs},
  journal = {Pacific Journal of Mathematics},
  volume  = {6},
  number  = {1},
  pages   = {1--8},
  year    = {1956}
}

@article{chen2009,
  author  = {Chen, Xi and Deng, Xiaotie and Teng, Shang-Hua},
  title   = {Settling the Complexity of Computing Two-Player Nash Equilibria},
  journal = {Journal of the ACM},
  volume  = {56},
  number  = {3},
  pages   = {Article 14},
  year    = {2009}
}

@article{daskalakis2009,
  author  = {Daskalakis, Constantinos and Goldberg, Paul W. and Papadimitriou, Christos H.},
  title   = {The Complexity of Computing a Nash Equilibrium},
  journal = {SIAM Journal on Computing},
  volume  = {39},
  number  = {1},
  pages   = {195--259},
  year    = {2009}
}

@article{hart2000,
  author  = {Hart, Sergiu and Mas-Colell, Andreu},
  title   = {A Simple Adaptive Procedure Leading to Correlated Equilibrium},
  journal = {Econometrica},
  volume  = {68},
  number  = {5},
  pages   = {1127--1150},
  year    = {2000}
}

\end{document}